\documentclass{article}
\usepackage{standalone}
\usepackage[colorlinks=true]{hyperref}
\usepackage{tikz}
\usepackage{enumerate}
\usepackage{amsthm,amsfonts,amssymb}
\usepackage{stmaryrd}
\usepackage{quotes}
\usepackage{subfig}
\usepackage{caption} 
\usepackage{xcomment}
\usepackage{xspace}
\usetikzlibrary{decorations.pathmorphing,decorations.pathreplacing,arrows,
automata,patterns}
\usepackage{complexity}
\usepackage{lineno}

\usepackage[
  style=alphabetic,
  firstinits=true,
  maxnames=6,
  minnames=5,
  maxbibnames=99,
  natbib=true,
  backend=bibtex
  ]{biblatex}
  
\def\ie{{\em i.e.}\xspace}

\newcommand{\ZZ}{\mathbb{Z}}

\newcommand{\NN}{\mathbb{N}}

\newtheorem{theorem}{Theorem}
\newtheorem{definition}{Definition}

\newtheorem{lemma}{Lemma}[section]
\newtheorem{corollary}[theorem]{Corollary}

\newcommand{\pizu}{$\Pi^0_1$\xspace}
\newcommand{\cantor}{\ensuremath{\left\{0,1\right\}^\NN}}
\newcommand{\zol}{\ensuremath{\left\{0,1\right\}^*}}
\newcommand{\turinf}{\leq_T}
\newcommand{\tursup}{\geq_T}
\newcommand{\turequiv}{\equiv_T}
\newcommand{\turdeg}{\deg_T}

\newcommand{\minimal}{{\ensuremath{\mathcal M}}}
\newcommand{\removable}{{\ensuremath{\mathcal C}}\xspace}

\begin{document}

\title{Turing degree spectra of minimal subshifts}
\author{Mike Hochman and Pascal Vanier\thanks{This work was        partially
 supported by grants EQINOCS ANR 11 BS02 004 03, TARMAC ANR 12 BS02 007 01 and   ISF grant 1409/11.}}

\maketitle

\begin{abstract}
Subshifts are shift invariant closed subsets of $\Sigma^{\ZZ^d}$, minimal 
subshifts are subshifts in which all points contain the same patterns. It has
been proved by Jeandel and Vanier that the Turing degree spectra of 
non-periodic minimal
subshifts always contain the cone of Turing degrees above any of its degree. It
was however not known whether each minimal subshift's spectrum was formed of exactly one cone or not. We construct inductively a minimal subshift whose spectrum consists of an uncountable number of cones with disjoint base.
\end{abstract}

A $\ZZ^d$-subshift is a closed shift invariant subset of $\Sigma^{\ZZ^d}$.
Subshifts may be seen as sets of colorings of $\ZZ^d$, with a finite number
of colors, avoiding some set of forbidden patterns. Minimal subshifts are subshifts containing no proper subshift, or equivalently subshifts in which 
all configurations have the same patterns. They are fundamental in the sense
that all subshifts contain at least a minimal subshift.

Degrees of unsolvability of subshifts have now been studied for a few years,
\citet{CDK2008,CDTW2010,CDTW2012} studied computability of one dimensional subshifts and proved some results about their Turing degree spectra: the \emph{Turing degree spectrum} of a subshift is the set of Turing degrees of its points, see \citet{KL2010}. \citet{Sim2011}, building on the work of \citet{Han1974,Mye1974}, noticed that the Medvedev and Muchnik degrees
of subshifts of finite type (SFTs) are the same as the Medvedev degrees of 
\pizu classes: \emph{\pizu classes} are the subsets of \cantor for which there exists a Turing machine halting only on oracles not in the subset.

Subsequently \citet{JeandelV2013:turdeg} focused on Turing degree spectra 
of different classes of multidimensional subshifts: SFTs, sofic and effective subshifts. They proved in particular that the Turing degree 
spectra of SFTs are almost the same as the spectra of \pizu classes: adding
a computable point to the spectrum of any \pizu class, one can construct an 
SFT with this spectrum. In order to prove that one cannot get a stronger
statement, they proved that the spectrum of any non-periodic 
minimal subshift contains the cone above any of its degrees:

\begin{theorem}[\citet{JeandelV2013:turdeg}]
    Let $X$ be a minimal non-finite subshift (\ie non-periodic in 
    at least one direction). For any point $x\in X$ and 
    any degree $\mathbf d\tursup \turdeg x$, there exists a point $y\in X$ 
    such that $\mathbf d=\turdeg y$.
\end{theorem}

Minimal subshifts are in particular interesting since any subshift contains
a minimal subshift \cite{Bir1912}. Here we answer the followup question 
of whether a 
minimal subshift always corresponds to a single cone or if there 
exists one with at least two cones of disjoint base. It is quite easy to
prove the following theorem:

\begin{theorem}
 For any Turing degree $\turdeg d$, there exists a minimal subshift $X$ such 
that the set of 
 Turing degrees of the points of $X$ is a cone of base $\turdeg d$.
\end{theorem}
For instance the spectrum of a Sturmian subshift \cite{MH1940} with an irrational angle 
is the cone whose base is the degree of the angle of the rotation. The 
theorem can also be seen as a corollary of Miller's 
proof~\cite{Mil2012}[Proposition 3.1] of a result on Medvedev degrees. 

In this paper, we prove the following result:

\begin{theorem}\label{thm:existence}
    There exist a minimal subshift $X\subset\{0,1\}^\ZZ$ and points $x_z\in X$ with $z\in \cantor$ such that for any $z\neq z'\in\cantor$, $\turdeg x_z$ and $\turdeg x_{z'}$ are incomparable and  such that there exists no point $y\in X$ with 
$\turdeg y\turinf\turdeg x_z$ and
$\turdeg y\turinf \turdeg x_{z'}$.
\end{theorem}

The subshift constructed in this proof is not effective and cannot be "effectivized", since minimal effective subshifts always contain 
a computable point and thus their spectra are the whole set of 
Turing degrees when they are non-periodic.

\section{Preliminary definitions}\label{S:preliminaries}

We give here some standard definitions and facts about
subshifts, one may consult the book of \citet{LindMarkus} 
for more details.

Let $\Sigma$ be a finite alphabet, its elements are called \emph{symbols}, the 
$d$-dimensional
full shift on $\Sigma$ is the set
$\Sigma^{\ZZ^d}$ of all maps (colorings) from $\ZZ^d$ to the $\Sigma$ (the 
colors). For
$v\in\ZZ^d$, the shift functions
$\sigma_v:\Sigma^{\ZZ^d}\to\Sigma^{\ZZ^d}$, are defined locally by
$\sigma_v(c_x)=c_{x+v}$. The full shift equipped with the distance
$d(x,y)=2^{-\min\left\{\left\|v\right\|\middle\vert v\in\ZZ^d,x_v\neq
y_v\right\}}$ is a compact metric space on which the shift functions act
as homeomorphisms. An element of $\Sigma^{\ZZ^d}$ is called a
\emph{configuration}.

Every closed shift-invariant (invariant by application of any $\sigma_v$)
subset $X$ of $\Sigma^{\ZZ^d}$ is called a \emph{subshift}. An element of a
subshift is called a \emph{point} of this subshift.

Alternatively, subshifts can be defined with the help of forbidden patterns. A
\emph{pattern} is a function $p:P \to \Sigma$, where $P$, the \emph{support}, is a finite subset of $\ZZ^d$. We say that a configuration
$x$ contains a pattern $p:P\to Sigma$, or equivalently that the pattern 
$p$ appears in $x$, if there exists $z\in\ZZ^d$ such that
$x_{|z+P}=p$.

Let $\mathcal F$ be a collection of \emph{forbidden} patterns, the
subset $X_F$ of $\Sigma^{\ZZ^d}$ containing the configurations having nowhere a
pattern of $F$. More formally, $X_{\mathcal F}$ is defined by

$$X_{\mathcal F}=\left\{x\in\Sigma^{\ZZ^d}\middle\vert \forall z\in\ZZ^d,\forall
p\in F, x_{|z+P}\neq p \right\}\textrm{.}$$

In particular, a subshift is said to be a \emph{subshift of finite type} (SFT)
when it can be defined by a collection of forbidden patterns that is finite. Similarly, an \emph{effective subshift} is a subshift which can be defined by
a recursively enumerable collection of forbidden patterns. A subshift is 
\emph{sofic} if it is the image of an SFT by a letter by letter function.

\begin{definition}[Minimal subshift]
 A subshift $X$ is called \emph{minimal} if it verifies one of the following 
 equivalent conditions:
 \begin{itemize}
  \item There is no subshift $Y$ such that $Y\subsetneq X$.
  \item All the points of $X$ contain the same patterns.
  \item It is the closure of the orbit of any of its points.
 \end{itemize}
\end{definition}
We will use the two latter conditions. 

For $x,y\in\cantor$, we say that $x\turinf y$ if there exists a Turing machine 
$M$ such that $M$ with oracle $y$ computes $x$. Of course $x\turequiv y$ when we
have both $x\turinf y$ and $y\turinf x$. The Turing degree of $x$ is the 
equivalence class of $x$ with respect to $\turequiv$. 
More details can be found in Rogers~\cite{Rogers}. We call \emph{recursive operator} a partial function $\phi:\cantor\to\cantor$ corresponding to a Turing
machine whose input is its oracle and output is an infinite sequence. We say
that the function is undefined on the inputs on which the Turing machine does not output an infinite sequence of bits.

For a possibly infinite word $w=w_0\dots w_n$, we denote $w_{[i,j]}=w_i\dots w_j$.

\section{Minimal subshifts with several cones}

\begin{lemma}\label{lem:twowords}
There exists a countable set $\removable\subseteq\cantor$ such that for any two recursive partial operators $\phi_1,\phi_2:\cantor\to\cantor$ and two distinct words 
$L=\{w_1,w_2\}\subseteq\zol$.  There exist two words $w_1',w_2'\in L^*$ starting 
respectively with $w_1$ and $w_2$
 such that we have one of the following:
 \begin{enumerate}[(a)]
  \item \label{lem:twowords:case:nonrec} either for any pair $x,y\in\cantor$, $\phi_1(w_1'x)$ differs from $\phi_2(w_2'y)$ when they are both defined,
  \item \label{lem:twowords:case:rec} or for any pair $x,y\in\cantor$, $\phi_1(w_1'x)=\phi_2(w_2'y)\in \removable$ when both defined.
 \end{enumerate}
\end{lemma}
\begin{proof}
 Let $M_1,M_2$ be the Turing machines computing the functions 
$x\mapsto\phi_1(w_1x),x\mapsto\phi_2(w_2x)$
 respectively. When restricting ourselves to inputs on which both operators are defined, it is quite clear that: 
 \begin{itemize}
  \item either there exists some sequences $x,y\in L^*$ such that $M_1(x)$'s 
    output differs from $M_2(y)$'s output at some step, 
  \item or the outputs of both machines $M_1,M_2$ do not depend on their
      inputs on their respective domains and are equal, in this latter case, 
      we are in case~\ref{lem:twowords:case:rec}. We define $\removable$ 
      to be the set of these outputs.
 \end{itemize}
 In the former case, there exist prefixes $w_1'$ and $w_2'$ of $w_1x$ and $w_2y$ 
 such that the partial outputs of $M_1$ once it has read $w_1'$ already 
 differs from the partial output of $M_2$ once it has read $w_2'$. 

 In the latter case, one may take $w_1'=w_1$ and $w_2'=w_2$. \removable is 
 countable since there is a countable number of quadruples 
 $\phi_1,\phi_2,w_1,w_2$.
\end{proof}

\begin{theorem}
 There exists a minimal subshift $X\subseteq \cantor$ whose set of Turing degrees contains 
$2^{\aleph_0}$ disjoint cones of Turing degrees.
\end{theorem}

Note that the following proof is in no way effective. As a matter of fact, 
all effective minimal subshifts contain some recursive point~\cite{BJ2010}, 
and their set of Turing degrees is the cone of \emph{all} degrees.

\begin{proof}

We construct a sequence of sofic subshifts $(X_i)_{i\in\NN}$ such that $X_{i+1}\subseteq 
X_i$ and such that the limit 
$X =\bigcap_{i\in\NN} X_i$ is minimal. In the process of constructing the 
$X_i$, which will be formed of
concatenations of allowed words $x^i_1,\dots ,x^i_k$, we ensure that no 
extensions of two distinct words 
may compute an identical sequence with any of the first $i$ Turing machines. At 
the same time, we make sure
that all allowed words of level $i+1$ contain all words of level $i$, thus 
enforcing the minimality of the
limit $X$. We also have to avoid that the limit $X$ contains a computable point.

Let $(\minimal_i)_{i\in\NN}$ be an enumeration of all minimal subshifts 
containing a point of the set \removable defined in lemma~\ref{lem:twowords}.
Such
an enumeration exists since \removable is countable and minimal subshifts are the closure of any of their points. We will also need an enumeration 
$(\phi_i)_{i\in\NN}$ of the 
partial recursive operators from \cantor to \cantor. 

Now let us define the sequence of sofic shifts $(X_i)_{i\in\NN}$. 
Each of these subshifts
will be the shift invariant closure of the biinfinite words formed by 
concatenations
of words of some language $L_i$ which here will be finite languages. We define
$X_0=\cantor$ that is to say $X_0$ is generated by $L_0=\{w_0=0,w_1=1\}$. Let 
us now give
the induction step. At each step, $L_{i+1}$ will contain $2^{i+1}$ words 
$w_{0\dots0},\dots,
w_{1\dots1}$, the indices being binary words of length $i+1$, which will verify 
the following conditions:
\begin{enumerate}
 \item\label{cond:start} The words $w_{b0},w_{b1}$ of $L_{i+1}$ start with the 
word $w_b$ of $L_i$.
 \item\label{cond:all} The words $w_b$ with $b\in\{0,1\}^{i+1}$ of $L_{i+1}$ 
each contain
 all the words of $w_{b'}$ with $b'\in\{0,1\}^i$ of $L_i$.
 \item\label{cond:mutuallower} For any two words $w_b,w_{b'}$ of $L_{i+1}$ and for all $j,j'\leq i$:
\begin{itemize}
    \item Either for  all $x,y\in L_i^\omega$, 
        $\phi_j(w_bx)\neq\phi_{j'}(w_{b'}y)$ when both defined,
    \item Or for  all $x,y\in L_i^\omega$, $\phi_j(w_bx),\phi_{j'}(w_{b'}y)$ are in \removable when defined.
\end{itemize}
 \item\label{cond:avoid} The words $w_{b0},w_{b1}$ do not appear in any 
configuration of $\minimal_j$,
 for all $j\leq i$.
\end{enumerate}

Conditions~\ref{cond:start} and~\ref{cond:all} are easy to ensure: just make 
$w_{ba}$ start with $w'=w_{b}w_{0\dots0}\dots w_{1\dots1}$. We then use 
Lemma~\ref{lem:twowords} to extend $w'$ into a word $w''$ verifying 
condition~\ref{cond:mutuallower}, this is done several times, once for 
every quadruple $w,w',j,j'$.
And finally, since $X_i$ is not minimal, we can extend $w''$ so that it
contains a pattern appearing in none of the $\minimal_j$'s for $j\leq i$, to 
obtain 
condition~\label{cond:avoid}. Now we can extend $w''$ with two different 
words thus obtaining $w_{b0}$ and $w_{b1}$.

Now let's check that this leads to the desired result: 
\begin{itemize}
 \item $X=\bigcap X_i$ is a
countable intersection of compact shift-invariant spaces, it is compact and
shift-invariant, thus a subshift. 
 \item Any pattern $p$ appearing in some point 
of $X$ is contained in a pattern $w_b$, with $b\in\{0,1\}^i$ for some $i$,
by construction (condition~\ref{cond:all}), all $w_{b'}$ with 
$b\in\{0,1\}^{i+1}$ contain $w_b$. 
Therefore, all points of $X$, since they are contained in $X_{i+1}$, contain 
$w_b$ and hence $p$. So $X$ is minimal.
 \item For all $z\in\cantor$, define the points $x_z=\lim_{i\to\infty} 
  x_{z[0,i]}$, they are in $X$ because they belong to each $X_i$. Condition~\ref{cond:mutuallower} ensures that if two of them compute the same 
  sequence $y\in\cantor$, then this sequence is in \removable. And 
condition~\ref{cond:avoid} ensures that no point of $X$ belongs to a minimal subshift containing a point of \removable.
\end{itemize}

\end{proof}

It is quite straightforward to transform this proof in order to get a subshift
on $\{0,1\}^\ZZ$ instead of \cantor and obtain the following corollary:

\begin{corollary}
    For any dimension $d$, there exists a minimal subshift 
    $X\subseteq\{0,1\}^{\ZZ^d}$ whose set of Turing degrees contains 
    $2^{\aleph_0}$ disjoint cones of Turing degrees.
\end{corollary}
\printbibliography
\end{document}